\documentclass{amsart}

\usepackage{amsfonts}
 \usepackage{mathabx}
\newtheorem{thm}{Theorem}[section]
\newtheorem{lemma}[thm]{Lemma}
\newtheorem{cor}[thm]{Corollary}

\theoremstyle{definition}
\newtheorem{definition}{Definition}

\theoremstyle{remark}
\newtheorem{remark}{Remark}

\numberwithin{equation}{section}

\newcommand{\abs}[1]{\lvert#1\rvert}

\begin{document}

\title{Large deviations of the Lyapunov exponent and localization for the 1D Anderson model}

 \author{Svetlana Jitomirskaya, Xiaowen Zhu}
 \address{Department of Mathematics, University of California, Irvine}

\begin{abstract}
The proof of Anderson localization for the 1D Anderson model with
arbitrary (e.g. Bernoulli) disorder, originally given by
Carmona-Klein-Martinelli in 1987, is based in part on the  multi-scale
analysis.  Later, in the 90s, it was realized that for one-dimensional models with positive Lyapunov exponents some parts of multi-scale analysis can be replaced by considerations involving subharmonicity and large deviation estimates for the corresponding cocycle, leading to nonperturbative proofs for 1D quasiperiodic models.  In this paper we present a short proof along these lines, for the Anderson model.  To prove dynamical localization we also develop a  uniform version of Craig-Simon's bound that works in high generality and may be of independent interest.
  \end{abstract}

\maketitle

\section{Introduction}

Anderson localization for the Anderson model can be proved in several different ways if the common distribution of the i.i.d.r.v's is absolutely continuous.  Without  that condition (or at least some H\"older regularity) it remains an open question for $d\geq 2,$ and the number of approaches that work for $d=1$ also drops dramatically. Such is the situation, for example, for the Bernoulli-Anderson model.  Anderson localization for {\it arbitrary} 1D disorder was first proved in \cite{ckm}. The approach was based on  certain regularity of the Lyapunov exponents coming from the (analysis around) the Furstenberg theorem to obtain an analogue of Wegner's lemma (automatic in the absolutely continuous case). After that  the proof was reduced  to multi-scale analysis, with initial scale coming again from the positive Lyapunov exponent.  Another argument was later presented in \cite{svw}, where an approach to positivity and regularity of the Lyapunov exponent using replica trick was given, again reducing the proof to multi-scale analysis.  Multi-scale analysis is a method that allows to achieve Green's function decay and ultimately localization from high probability of decay at the initial scale.  It works in a variety of settings. Originally developed by Frohlich and Spencer \cite{fs}, it was significantly simplified in \cite{von1989new} but remains somewhat involved. It should be noted  that in the multidimensional case no shortcuts such as Furstenberg theorem or replica trick are available, and the multi-scale analysis is used to reach conclusions analogous to the positivity of the Lyapunov exponent simultaneously with the proof of localization. Yet in the one-dimensional case positivity of the Lyapunov exponent essentially provides the averaged decay statement, thus a large portion of the {\it conclusion} of the multi-scale analysis, making its machinery seem redundant.

A method to effectively exploit  positive Lyapunov exponent for a localization proof based on the analysis of the large deviation set for the Lyapunov exponent was first developed in \cite{j} for the almost Mathieu operator, initiating what was later called a non-perturbative approach, in contrast with earlier proofs based on some form of multi-scale analysis \cite{fsw,sin}. A robust method  based on subharmonic function theory and the theory of semianalytic sets was then developed in \cite{bg} and other papers summarized in \cite{bbook}, to conclude localization from positive Lyapunov exponents for analytic quasiperiodic and some other deterministic potentials. The fact that those ideas can be applicable also to the Anderson model was mentioned in some talks by one of the authors circa 2000, but the details were never developed. One goal of this paper is to obtain a proof of Anderson localization for the 1D Anderson model in the spirit of \cite{j} but with appropriate simplifications due to randomness.

Another proof, also based on large deviations and also avoiding multi-scale analysis was recently developed in \cite{7}. The proof of \cite{7} is based on deterministic  ideas close to the ones in \cite{bs}, which we believe may be somewhat more complicated than needed for the random case. We mention that yet another, purely dynamical,  proof of localization for the 1D Anderson model appears in \cite{gk}.

One ingredient in our simple argument for spectral localization,
%in Sec. \ref{alllemma},
Theorem \ref{CS}, is Craig-Simon's upper bound based on subharmonicity of the Lyapunov exponent \cite{cs}, a statement that holds for any ergodic potential. In order to prove dynamical localization  we need a uniform in energy and quantitative version of this statement, that we prove for general ergodic potentials satisfying certain large deviation bounds,  a result that could be of independent interest. We note that our proof does not explicitly use subharmonicity.

The rest of this paper is organized as follows. Section \ref{pre} contains the preliminaries, the statement of the spectral localization result, Theorem \ref{thm1}, and its quick reduction to Theorem \ref{thm2}.  We then prove the preparatory Lemmas \ref{lemma1}, \ref{omega1}, \ref{omega3}, and Corollary \ref{omega2} in Section \ref{alllemma}. Then we complete the proof of Theorem \ref{thm2} in Section \ref{pf}. Our proof effectively establishes a more precise result, Theorem \ref{thm22}, which in turn immediately implies the Lyapunov behavior at all eigenvalues, Theorem \ref{CS2}. We formulate and prove the general uniform Craig-Simon-type statement in Section \ref{uniformcs}, and use it in Section \ref{dyn} to prove dynamical localization.
\section{Preliminaries}\label{pre}

The one dimensional Anderson model is given by a discrete
Schr\"odinger operators $H_{\omega}$
%with real $i.i.d.$ potential $\{V_{\omega}(n)\}$:
\begin{equation}
  (H_{\omega}\Psi)(n)=\Psi(n+1)+\Psi(n-1)+\omega_n\Psi(n),
\end{equation}
where $\omega_n\in\mathbb{R}$ are  independent identically distributed random variables with common
Borel probability distribution $\mu.$ We will assume that $S\subset\mathbb{R}$, the topological
support of $\mu,$ is  compact , and contains at least two points. We will
denote the probability space $\Omega=S^{\mathbb{Z}},$ with elements
$\{\omega_n\}_{n\in\mathbb{Z}}\in\Omega$.
%, $i.e.$ for each
%$n\in\mathbb{Z}$ and will think of  $V_{\omega} (n)$ as
%random variable depending on $\omega_n \in S$.
%We will consider $V_\omega$ in the product probability space $
%(S^{\mathbb{Z}},\mu^{\mathbb{Z}})$ as a whole instead.
Denote $\mu^{\mathbb{Z}}$ as $\mathbb{P}$. Let $\mathbb{P}_{[a,b]}$ be
$\mu^{[a,b]\cap \mathbb{Z}}$ on $S^{[a,b]\cap \mathbb{Z}}$. Aldo let
$T$ be the shift  $T\omega_i=\omega_{i-1}$.  Finally, we denote
Lebesgue measure on $\mathbb{R}$ by $m$. % We know that $\sigma(H_\omega)=\sigma(H)=[-2,2]+S$ for $a.e. \omega$.
We say that $H_\omega$ has  spectral localization  in $I$ if for
a.e. $\omega$, $H_\omega$ has only pure point spectrum in $I$ and its
eigenfunctions $\Psi(n)$ decay exponentially in $n$.
%We are going to give a new proof for Anderson model based on the large deviation estimates and subharmonicity of Lyapunov exponents.
  \begin{definition}  We call $E$ a generalized eigenvalue ($g.e.$), if there exists a nonzero polynomially bounded function $\Psi(n)$ such that $H_\omega\Psi=E\Psi$. We call $\Psi(n)$ a generalized eigenfunction.
\end{definition}
Since the set of g.e. supports the spectral measure of $H_\omega$
(e.g. \cite{cycon}), we only need to show:
\begin{thm}\label{thm1}
  For a.e. $\omega$, for every g.e. $E,$ the corresponding generalized eigenfunction $\Psi_{\omega,E}(n)$ decays exponentially in $n$.
\end{thm}

For $[a,b]$ an interval, $a,b\in\mathbb{Z}$, define $H_{[a,b],\omega}$ to be operator $H_\omega$ resticted to $[a,b]$ with zero boundary conditions outside $[a,b]$. Note that it can be expressed as a "$b-a+1$"-dimensional matrix.
The Green's function for $H_\omega$ restricted to $[a,b]$ with energy $E\notin\sigma_{[a,b],\omega}$ is
  \[
    G_{[a,b],E,\omega}=(H_{[a,b],\omega}-E)^{-1}
  \]
Note that this can also be expressed as a "$b-a+1$"-dimensional matrix. Denote its $(x,y)$ entry as $G_{[a,b],E,\omega}(x,y)$.

It is well known that
  \begin{equation}\label{possion}
    \Psi(x)=-G_{[a,b],E,\omega}(x,a)\Psi(a-1)-G_{[a,b],E,\omega}(x,b)\Psi(b+1),\quad x\in[a,b]
  \end{equation}
and we have
\begin{equation}\label{sigma}
\sigma:=\sigma(H_{\omega})=[-2,2]+S\quad a.e. \omega.
\end{equation}

% If one can get that, the Green's function near $n$, say, for example on $[n-k,n+k]$, is decaying somehow exponentially in $n$ as $n$ growing, then since $\Psi$ on the  right-hand-side is polynomially bounded, $\Psi(n)$ on the left-hand-side will decay exponentially in $n$, too.

% inspires us to define "regular and singular".

% \begin{definition}
% $\mathcal{A}_{[a,b]}=\{[a,b],[a,b-1],[a+1,b],[a+1,b-1]\}$
% \end{definition}
\begin{definition}
   For $c>0, n\in\mathbb{Z}$, we say $x\in\mathbb{Z}~$ is $(c,n,E,\omega)$-regular, if
  \[
    G_{[x-n,x+n],E,\omega}(x,x-n) \leq e^{-cn}
  \]
  % i.e. if $A=[x_1,x_2]$
  \[
    G_{[x-n,x+n],E,\omega}(x,x+n) \leq e^{-cn}
  \]
  Otherwise, we call it $(c,n,E,\omega)$-singular.
\end{definition}

By \eqref{possion} and definition 2, Theorem  \ref{thm1} follows from
\begin{thm}\label{thm2}
  There exists $ \Omega_0$ with $\mathbb{P}(\Omega_0)=1$, such that
  for every $ \tilde{\omega}\in\Omega_0$, for any g.e. $\tilde{E}$ of
  $H_{\tilde{\omega}}$, there exist $
  N=N(\tilde{E},\tilde{\omega}),C=C(\tilde{E})$, such that for every $ n>N$, $2n,~2n+1$ are $(C,n,\tilde{E},\tilde{\omega})$-regular.
\end{thm}
% \begin{remark}
%   It's similar for even terms. We omit them only because of notation reasons.
% \end{remark}
% \begin{remark}
%   Because if we achieve this, denote the corresponding polynomially bounded generalized eigenfunction as $\Psi(n)=\Psi_{\tilde{\omega},\tilde{E}}(n)\leq M(1+n)^p$, for $p>0$. Then for every $ n>N$, let $A_n=[n+1,3n+1]$
%   \[
%   \left\vert G_{A_n,E,\omega}(x,\partial A_n)\right\vert \leq e^{-C\abs{x-\partial A_n}}
%   \]
%   so by eqution \ref{possion}, and $\abs{x-\partial A_n}\geq n-1$.
%   \[
%   \abs{\Psi(2n+1)}\leq Me^{-Cn}(1+3n+1)^p
%   \]
%   So for large enough $n$, $\Psi(2n+1)$ decays exponentially in $n$. Similarly for even terms and we will get Theorem \ref{thm1}.
% \end{remark}

Some other standard basic settings are below. Denote
\[
  P_{[a,b],E,\omega}=det(H_{[a,b],E,\omega}-E), a\leq b
\]
If $a>b$, let $  P_{[a,b],E,\omega}=1$. Then
\begin{equation}\label{A}
  \left\vert G_{[a,b],E,\omega}(x,y)\right\vert=\frac{\left\vert P_{[a,x-1],E,\omega}P_{[y+1,b],E,\omega}\right\vert}{\left\vert P_{[a,b],E,\omega}\right\vert},\quad x\leq y
\end{equation}
%   \[
%   \left\vert G_{[a,b](x,y),E,\omega}\right\vert=\frac{\left\vert P_{[a,y-1],E,\omega}P_{[x+1,b],E,\omega}\right\vert}{\left\vert P_{[a,b],E,\omega}\right\vert},\quad x\geq y?
% \]
If we denote the transfer matrix $T_{[a,b],E,\omega}$ as the matrix such that
\[
\left(
\begin{array}{c}
  \Psi(b)\\
  \Psi(b-1)\\
\end{array}
\right)
=T_{[a,b],E,\omega} \left(
\begin{array}{c}
  \Psi(a)\\
  \Psi(a-1)\\
\end{array}
\right)
\]
where $\Psi$ solves $H_\omega\Psi=E\Psi,$ then
\[
T_{[a,b],E,\omega}=\left(
  \begin{array}{cc}
    P_{[a,b],E,\omega} & -P_{[a+1,b],E,\omega}\\
    P_{[a,b-1],E,\omega} & -P_{[a+1,b-1],E,\omega}\\
  \end{array}
  \right)
\]
The Lyapunov exponent exists by Kingman's subadditive ergodic theorem and is given by
  \begin{equation}\label{gamma}
    \gamma(E)=\lim_{n\to\infty}\frac{1}{n}\int_0^1 \log\Vert
    T_{[0,n],E,\omega}\Vert
    d\mathbb{P}(\omega)=\lim_{n\rightarrow\infty}\frac{1}{n} \log\Vert
    T_{[0,n],E,\omega}\Vert, \quad a.e. \omega.
\end{equation}

Let $\nu=\inf\limits_{E\in \sigma}\gamma(E)$. By the Furstenberg's
theorem $\nu>0.$ It follows from (\ref{A}) that the desired
  exponential decay of the Green's function can be achieved if all the $P_{[a,b]}$ in  (\ref{A}) behave as $e^{(b-a)\gamma(E)}$, thus leading
    to the study of deviations of $\ln P_{[a,b]} $ from its mean. In
    fact, the key estimates underlying the analysis of \cite{ckm}
are precisely  large deviation bounds for the Lyapunov exponent due to
Le Page \cite{lepage}. Here we will use a corresponding statement for
the matrix elements \cite{tsay1999some}

\begin{lemma}[ ''uniform-LDT'']\label{ldt lemma}
  For any $\epsilon>0$, there exists $\eta=\eta(\epsilon)>0$ such that, there exists $N_0=N_0(\epsilon)$, such that for every $ b-a>N_0$, and any $E$ in a compact set,
  \begin{equation}\label{eqLDT}
  \mathbb{P}\left\{ \omega:\left\vert \frac{1}{b-a+1} \log\Vert P_{[a,b],E,\omega}\Vert-\gamma(E) \right\vert\geq\epsilon
   \right\} \leq e^{-\eta (b-a+1)}
 \end{equation}
\end{lemma}

It will also be convenient to use the general subharmonicity upper
bound due to Craig-Simon \cite{cs}

\begin{thm}[Craig-Simon \cite{cs}]\label{CS1}
  For a.e. $\omega$ {\it for all} $E$, we have
  \begin{equation}
    % \begin{aligned}
   \varlimsup_{n\to\infty} \frac{\log\Vert T_{[0,n],E,\omega}\Vert}{n+1}\leq\gamma(E)
  \end{equation}

\end{thm}

\section{Main lemmas}\label{alllemma}

Denote
\begin{equation}\label{B+}
    B_{[a,b],\epsilon}^{+} =\left\{(E,\omega): |P_{[a,b],E,\omega}|\geq e^{(\gamma(E)+\epsilon)(b-a+1)}\right\}
\end{equation}
\begin{equation}\label{B-}
    B_{[a,b],\epsilon}^{-} =\left\{(E,\omega): |P_{[a,b],E,\omega}|\leq e^{(\gamma(E)-\epsilon)(b-a+1)}\right\}\
\end{equation}
and denote $B_{[a,b],\epsilon,E}^{\pm}=\{\omega:(E,\omega)\in B_{[a,b],\epsilon}^{\pm}\}$, $B_{[a,b],\epsilon,\omega}^{\pm}=\{E:(E,\omega)\in B_{[a,b],\epsilon}^{\pm}\}$,
$B_{[a,b],*}=B_{[a,b],*}^+\cup B_{[a,b],*}^-$.

Let $E_{j,(\omega_a,\cdots,\omega_b)}$ be eigenvalues of $H_{[a,b],\omega}$ with $\omega|_{[a,b]}=(\omega_a,\cdots,\omega_b)$.

Large deviation theorem gives us the estimate that for all $E, a, b,\epsilon$
\begin{equation}\label{ldt}
\mathbb{P}(B_{[a,b],\epsilon,E}^\pm)\leq e^{-\eta(b-a+1)}
\end{equation}

Assume $\epsilon=\epsilon_0<\frac{1}{8}\nu$ is fixed for now, so we
omit it from the notations until Lemma \ref{omega3}. Let $\eta_0=\eta(\epsilon_0)$ be the corresponding parameter from Lemma \ref{ldt lemma}
\begin{lemma}\label{lemma1}
For $n \geq 2$, if $x$ is $(\gamma(E)-8\epsilon_0,n,E,\omega)$-singular, then
\[(E,\omega)\in B_{[x-n,x+n]}^-\cup B_{[x-n,x]}^+\cup B_{[x,x+n]}^+
\]
\end{lemma}
\begin{remark}
  Note that from \eqref{ldt}, for all $E,x,n\geq 2$,
  \[
    \mathbb{P}(B_{[x-n,x+n],E}^-\cup B_{[x-n,x],E}^+\cup B_{[x,x+n],E}^+)\leq 3e^{-\eta_0 (n+1)}
  \]
\end{remark}
\begin{proof}
Follows imediately from the definition of singularity and \eqref{A}.
% Assume not, then
% \[
%   \left\{
%   \begin{aligned}
%      & |P_{[x-n,x+n],E,\omega}|\geq e^{(\gamma(E)+\epsilon_0)(2n+1)} \\
%      & |P_{[x-n,x],E,\omega}|\leq e^{(\gamma(E)-\epsilon_0)(n+1)} \\
%      & |P_{[x,x+n],E,\omega}|\leq e^{(\gamma(E)-\epsilon_0)(n+1)}
%   \end{aligned}
%   \right.
% \]
% So we can estimate
% \[
%   \begin{aligned}
%     \left\vert G_{[x-n,x+n],E,\omega}(x,x-n)\right\vert
%     &=  \frac{\left\vert P_{[x,x+n],E,\omega}\right\vert}{\left\vert P_{[x-n,x+n],E,\omega}\right\vert}\\
%     &\leq  \frac{e^{(\gamma(E)+\epsilon_0)(n+1)}}{e^{(\gamma(E)-\epsilon_0)(2n+1)}}\\
%     &\leq  e^{-\gamma(E)(n)+\epsilon_0(3n+2)}\\
%     &\leq  e^{-(\gamma(E)-8\epsilon_0)n}
%   \end{aligned}
% \]
% Similar for $G_{[x-n,x+n],E,\omega}(x,x+n)$.
% Thus $x$ is $(\gamma(E))-8\epsilon_0, n, E,\omega)$-regular, contradiction.
\end{proof}

% By Theorem \ref{thm2},
% From lemma \ref{lemma1}, we get our first restriction on $\Omega$ to make the bad sets have small measure.
Now we will use the following three lemmas to find the proper $\Omega_0$ for Theorem \ref{thm2}.
\begin{lemma}\label{omega1}
  Let $0<\delta_0<\eta_0$. For a.e. $\omega$ (we denote this set as $\Omega_1$), there exists $ N_1=N_1(\omega)$, such that for every $ n>N_1$,
  \[
  \max\{m(B_{[n+1,3n+1],\omega}^-),m(B_{[-n,n],\omega}^-)\}\leq e^{-(\eta_0-\delta_0)(2n+1)}
  \]
\end{lemma}

\begin{proof}
By \eqref{ldt},
  % for every $ E\in I=\sigma$,which is compact, $P(B_{[n+1,3n+1],E}^-)\leq e^{-\eta_0(2n+1)}$ and $P(B_{[-n,n],E}^-)\leq e^{-\eta_0(2n+1)}$
\[
m\times\mathbb{P}(B_{[n+1,3n+1]}^-)\leq m(\sigma)e^{-\eta_0(2n+1)}
\]
\[
m\times\mathbb{P}(B_{[-n,n]}^-)\leq m(\sigma)e^{-\eta_0(2n+1)}
\]
If we denote
    \[
      \Omega_{\delta_0,n,+}=\left\{\omega:m(B_{[n+1,3n+1],\omega}^-)\leq e^{-(\eta_0-\delta_0)(2n+1)}\right\}
    \]
    \[
    \Omega_{\delta_0,n,-}=\left\{\omega:m(B_{[-n,n],\omega}^-)\leq e^{-(\eta_0-\delta_0)(2n+1)}\right\},
    \]
We have by Tchebyshev,
\begin{equation}\label{A5}
      \mathbb{P}(\Omega_{\delta_0,n,\pm}^c)
      % &\leq e^{(\eta_0-\delta_0)n}\int_{\Omega} m(B_{[n+1,3n+1],\omega})dP\omega\\
      % &=    e^{(\eta_0-\delta_0)n}\int_I P(B_{[n+1,3n+1],E})dx\\
      % &\leq e^{(\eta_0-\delta_0)n}m(I)e^{-\eta_0(2n+1)}\\
      \leq m(\sigma)e^{-\delta_0(2n+1)}.
\end{equation}
By Borel-Cantelli lemma, we get for $a.e.~\omega$,
\[
\max\{m(B_{[n+1,3n+1],\omega}^-),m(B_{[-n,n],\omega}^-)\}\leq e^{-(\eta_0-\delta_0)(2n+1)},
\]
for $n>N_1(\omega)$.
\end{proof}
\begin{remark}\label{N1}
  Note that we can actually shift the operator and use center point
  $l$ instead of $0$. Then we will get $\Omega_1(l)$ instead of
  $\Omega_1$, $N_1(l,\omega)$ instead of $N_1(\omega)$. And if we pick
  $N_1(l,\omega)$ in the theorem as the smallest integer satisfying
  the conclusion, we can estimate when we will have $N_1(l,\omega)\leq \ln^2 |l|$, which is very useful in the proof for dynamical localization in section 6.
  % In fact, $\mathbb{P}\{\omega: N_1(l,\omega)>\ln^2 |l|\}\leq C' e^{-\delta_0(2|\ln^2|l|+1)}$, By Borel-Cantelli, for $a.e.\omega$,(We denote this set as $\Omega_{N_1}$,) there exists $L_1(\omega)$, such that for any $|l|>L_1(\omega)$, $N_1(l,\omega)\leq \ln^2|l|$.
\end{remark}
The next results follows from :
\begin{thm}
%[Craig-Simon \cite{cs}]
\label{CS}
  For a.e. $\omega$(we denote this set as $\Omega_2$), for all $E$, we have
  \begin{equation}\label{C}
    % \begin{aligned}
      \max\left\{\varlimsup_{n\to\infty} \frac{ \log\Vert T_{[-n,0],E,\omega}\Vert}{n+1}, \varlimsup_{n\to\infty} \frac{\log\Vert T_{[0,n],E,\omega}\Vert}{n+1}\right\}\leq\gamma(E)
  \end{equation}
  \begin{equation}\label{D}
    \max\left\{\varlimsup_{n\to\infty} \frac{\log\Vert T_{[n+1,2n+1],E,\omega}\Vert}{n+1}, \varlimsup_{n\to\infty} \frac{\log\Vert T_{[2n+1,3n+1],E,\omega}\Vert}{n+1}\right\}\leq\gamma(E)
  \end{equation}
\end{thm}
\begin{remark}
  \eqref{C} is a direct reformulation of the result of \cite{cs},
  Theorem \ref{CS1}, while \eqref{D} follows by exactly the same proof.
\end{remark}
\begin{cor}\label{omega2}
  For every $ \omega\in\Omega_2$, for every $E$, there exists $N_2=N_2(\omega,E)$, such that for every $ n>N_2$,
\[
\begin{aligned}
&\max\{\Vert T_{[-n,0],E,\omega}\Vert, \Vert T_{[0,n],E,\omega}\Vert\}<e^{(\gamma(E)+\epsilon)(n+1)}\\ &\max\{\Vert T_{[n+1,2n+1],E,\omega}\Vert, \Vert T_{[2n+1,3n+1],E,\omega}\Vert\}<e^{(\gamma(E)+\epsilon)(n+1)}
\end{aligned}
\]

  % \[
  %   \begin{aligned}
  %     &\Vert T_{[-n,0],E,\omega}\Vert< e^{(\gamma(E)+\epsilon)(n+1)}\\
  %     &\Vert T_{[0,n],E,\omega}\Vert< e^{(\gamma(E)+\epsilon)(n+1)}\\
  %     &\Vert T_{[n+1,2n+1],E,\omega}\Vert< e^{(\gamma(E)+\epsilon)(n+1)}\\
  %     &\Vert T_{[2n+1,3n+1],E,\omega}\Vert< e^{(\gamma(E)+\epsilon)(n+1)}
  %   \end{aligned}
  % \]
\end{cor}
% \begin{remark}
%   Basically speaking, the only difference from Theorem 1.5 in \cite{craig1983subharmonicity} is that we are considering restrictions on some different box-sequences, for example $\{[n+1,2n+1]\}$, instead of the original boxes $\{[0,n]\}$. However, by\ref{LE}, $\gamma(E)$ keeps constant under $\{[n+1,2n+1]\}$, so subharmonic. While $\bar{\gamma_(E)}$ as limsup of $\gamma_(E)_[n+1,2n+1]$ is still submean since $\gamma_(E)_[n+1,2n+1]$ are submean. By properties of submean and subharmonic, together with Fustenberg Theorem and Fubini, we can get the results.
% \end{remark}
% \begin{proof}
%   Only prove $[n+1,2n+1]$ case.\\
%   Claim:
%   \begin{enumerate}
%     \item $\gamma(E)$ is subharmonic.
%     \item $\varlimsup_{n\to\infty} \frac{1}{n+1} \log\Vert T_{[n+1,2n+1],E,\omega}\Vert$ is submean
%   \end{enumerate}
% \end{proof}
\begin{lemma}\label{omega3}
   Let $\epsilon>0,K>1$, For a.e. $\omega$(we denote this set as $\Omega_3=\Omega_3(\epsilon,K)$), there exists $ N_3=N_3(\omega)$, so that for every $ n>N_3$, for every $ E_{j,(\omega_{n+1},\cdots,\omega_{3n+1})}$, for every $ y_1,y_2$ satisfying $-n\leq y_1\leq y_2\leq n$,  $\abs{-n-y_1}\geq\frac{n}{K}$, and $\abs{n-y_2}\geq\frac{n}{K}$,
 we have $E_{j,(\omega_{n+1},\cdots,\omega_{3n+1})}\notin B_{[-n,y_1],\epsilon,\omega}\cup B_{[y_2,n],\epsilon,\omega}$.
\end{lemma}
\begin{remark}
  Note that $\epsilon$ and $K$ are not fixed yet, we're going to determine them later in section \ref{pf}.
\end{remark}
\begin{proof}
% In order to use Borel-Cantelli, one need to estimate
Let $\bar {\mathbb{P}}$
% \[
%   \bar{P}=P\left(\bigcup\limits_{y_1,y_2}\bigcup\limits_{j=1}^{2n+1}
%   B_{[-n,y_1],\epsilon,E_{j,(\omega_a,\cdots,\omega_b)}\cup B_{[y_2,n],\epsilon,E_{j,(\omega_{n+1},\cdots,\omega_{3n+1})}\right)
% \]
be the probability that there are some $y_1, y_2, j$ with \[
E_{j,(\omega_{n+1},\cdots,\omega_{3n+1})}\in B_{[-n,y_1],\epsilon,\omega}\cup B_{[y_2,n],\epsilon,\omega}.
\] Note that for any fixed $\omega_c,\cdots,\omega_d $, with $[c,d]\cap[a,b]=\emptyset$, by independence,
\[
\mathbb{P}(B_{[a,b],\epsilon,E_{j,(\omega_c,\cdots,\omega_d)}})=\mathbb{P}_{[a,b]}(B_{[a,b],\epsilon,E_{j,(\omega_c,\cdots,\omega_d)}})\leq e^{-\eta_0(b-a+1)}
\]
Applying to $[a,b]=[-n,y_1]$ or $[y_2,n]$, $[c,d]=[n+1,3n+1]$ and integrating over $\omega_{-n},\cdots,\omega_{y_1}$ or $\omega_{y_2},\cdots,\omega_{n}$, we get
\[
\mathbb{P}(B_{[-n,y_1],\epsilon,E_{j,(\omega_{n+1},\cdots,\omega_{3n+1})}}\cup B_{[y_2,n],\epsilon,E_{j,(\omega_{n+1},\cdots,\omega_{3n+1})}}) \leq 2e^{-\eta_0(\frac{n}{K}+1)},
\]
so
\begin{equation}\label{B5}
\bar{\mathbb{P}}\leq(2n+1)^3 2e^{-\eta_0(\frac{n}{K}+1)}
\end{equation}
Thus by Borel-Cantelli, we get the result.

\end{proof}
\begin{remark}\label{N3}
 Similar to remark \ref{N1}, we can get $\Omega_3(l)$, $N_3(l,\omega)$
 for an operator shifted by $\ell$ instead, and get the result that
 for $a.e. \omega$ (we denote this set as $\Omega_{N_3}$),  there
 exists $L_3(\omega)$, such that for any $|l|>L_3$, $N_3(l,\omega)\leq
 \ln^2 |l|.$ This will be of use in section 6 for proving dynamical localization.
\end{remark}

\section{Proof of Theorem 2.2}\label{pf}
We will only provide a proof that $2n+1$ is $(c, n, E,\omega)$-regular, the argument for $2n$ being similar.
\begin{proof}
Let $\epsilon$ be small enough such that
  \begin{equation}\label{epsilon1}
    \epsilon<\min\{(\eta_0-\delta_0)/3,\nu\}.
  \end{equation}
Now let
  \[L:=e^{(\eta_0-\delta_0-\epsilon)}>1,\]
and note that since $S$ is bounded, by \eqref{sigma} we have there exists $ M>0$, such that
\[
|P_{[a,b],E,\omega}|<M^{(b-a+1)},\quad \forall E\in\sigma,\omega
\]
Pick $K$ big enough such that
  \[M^{\frac{1}{K}}<L\]
Let $\sigma>0$ be such that
\begin{equation}\label{K}
M^{\frac{1}{K}}\leq L-\sigma<L
\end{equation}
Let $\Omega_0=\Omega_1\cap\Omega_2\cap\Omega_3(\epsilon,K)$. Pick
$\tilde{\omega}\in\Omega_0$, and take $\tilde{E}$ a $g.e.$ for
$H_{\tilde{\omega}}$, with $\Psi$ the corresponding generalized eigenfunction.
Without loss of generality assume $\Psi(0)\neq 0$. Then there exists $ N_4$, such that for every $ n>N_4$, 0 is $(\gamma(\tilde{E})-8\epsilon_0,n,\tilde{E},\tilde{\omega})$-singular.

% Assume $2n+1$ is not eventually $(\gamma(\tilde{E})-8\epsilon_0,n,\tilde{E},\tilde{\omega})$-regular, then there exists $ \{n_k\}$ with $n_k\to\infty$, such that $2n_k+1$ is  $(\gamma(\tilde{E})-3\epsilon_0,n_k,\tilde{E},\tilde{\omega})$-singular.

For $n>N_0=\max\{N_1(\tilde{\omega}),N_2(\tilde{\omega},\tilde{E}),N_3(\tilde{\omega}),N_4(\tilde{\omega},\tilde{E})\}$, assume $2n+1$ is $(\gamma(\tilde{E})-8\epsilon_0,n,\tilde{E},\tilde{\omega})$-singular.
% WLOG use $\{n\}$ istead of $\{n_k\}$, we have that
% \begin{itemize}
 Then both $0$ and $2n+1$ is $(\gamma(\tilde{E})-8\epsilon_0,n,\tilde{E},\tilde{\omega})$-singular. So by Lemma \ref{lemma1},
  $\tilde{E}\in B_{[n+1,3n+1],\epsilon_0,\tilde{\omega}}^-\cup B_{[n+1,2n+1],\epsilon_0,\tilde{\omega}}^+ \cup B_{[2n+1,3n+1],\epsilon_0,\tilde{\omega}}^+ $.
  By Corollary \ref{omega2} and \eqref{B+}, $\tilde{E}\notin B_{[n+1,2n+1],\epsilon_0,\tilde{\omega}}^+\cup B_{[2n+1,3n+1],\epsilon_0,\tilde{\omega}}^+$, so it can only lie in $B_{[n+1,3n+1],\epsilon_0,\tilde{\omega}}^-$.

  % But  by \eqref{B-}.
  % \begin{equation}\label{bad}
  %   B_{[n+1,3n+1],\epsilon_0,\tilde{\omega}}^-= \left\{E:|P_{[n+1,3n+1],\epsilon,E,\tilde{\omega}}|\leq e^{(\gamma(E)-\epsilon_0)(2n+1)}\right\}
  % \end{equation}
   Note that in \eqref{B-}, $P_{[n+1,3n+1],E,\tilde{\omega}}$ is a
   polynomial in $E$ that has $2n+1$ real zeros (eigenvalues of
   $H_{[n+1,3n+1],\tilde{\omega}}$), which are all in
   $B=B^{-}_{[n+1,3n+1],\epsilon,\tilde{\omega}}$. Thus $B$ consists
   of less than or equal to $2n+1$ intervals around the eigenvalues. $\tilde{E}$ should lie in one of them. By Lemma \ref{omega1}, $m(B)\leq Ce^{-(\eta_0-\delta_0)(2n+1)}$. So there is some e.v. $E_{j,[n+1,3n+1],\tilde{\omega}}$ of $H_{[n+1,3n+1],\omega}$ such that
   \[
   \vert\tilde{E}-E_{j,[n+1,3n+1],\tilde{\omega}}\vert\leq e^{-(\eta_0-\delta_0)(2n+1)}
   \]
  By the same argument, there exists $ E_{i,[-n,n],\tilde{\omega}}$, such that
   \[
   \vert\tilde{E}-E_{i,[-n,n],\tilde{\omega}}\vert\leq e^{-(\eta_0-\delta_0)(2n+1)}
   \]
   Thus $\vert E_{i,[-n,n],\tilde{\omega}}-E_{j,[n+1,3n+1],\tilde{\omega}}\vert\leq 2e^{-(\eta_0-\delta_0)(2n+1)}$. However, by Theorem \ref{omega3}, one has $E_{j,[n+1,3n+1],\tilde{\omega}}\notin B_{[-n,n],\epsilon,\tilde{\omega}}$, while $E_{i,[-n,n],\tilde{\omega}}\in B_{[-n,n],\epsilon,\tilde{\omega}}$
   This will give us a contradiction below.\\
% \end{itemize}
% ~\\
Since $\vert E_{i,[-n,n],\tilde{\omega}}-E_{j,[n+1,3n+1],\tilde{\omega}}\vert\leq 2e^{-(\eta_0-\delta_0)(2n+1)}$ and $E_{i,[-n,n],\tilde{\omega}}$ is the e.v. of $H_{[-n,n],\tilde{\omega}}$,
\[
  \left\Vert G_{[-n,n],E_{j,[n+1,3n+1],\tilde{\omega}},\tilde{\omega}}\right\Vert\geq \frac{1}{2}e^{(\eta_0-\delta_0)(2n+1)}
\]
Thus there exist $ y_{1},y_{2}\in [-n,n]$ and such that
% Need fix
\[
  \left\vert G_{[-n,n],E_{j,[n+1,3n+1],\tilde{\omega}},\tilde{\omega}}(y_{1},y_{2})\right\vert\geq \frac{1}{2n}e^{(\eta_0-\delta_0)(2n+1)}
\]
Let $E_j=E_{j,[n+1,3n+1],\tilde{\omega}}$. We have $E_j\notin B_{[-n,n],\epsilon,\tilde{\omega}}$, thus
\[
\vert P_{[-n,n],\epsilon,E_j,\tilde{\omega}}\vert\geq e^{(\gamma(E_j)-\epsilon)(2n+1)}
\]
so by \eqref{A},
\begin{equation}\label{last}
  \left\Vert P_{[-n,y_{1}],\epsilon,E_jï¼Œ\tilde{\omega}}P_{[y_{2},n],\epsilon,E_j,\tilde{\omega}}\right\Vert\geq\frac{1}{2n}e^{(\eta_0-\delta_0)(2n+1)}e^{(\gamma(E_j)-\epsilon)(2n+1)}
\end{equation}
% Let $M= sup\{|V|+|E_j|+2\}$, where $|V|$ is assumed bounded, $E_i,E_j$ are bounded because they are close to $E\in I$.\\
% Then pick $\epsilon$ small enough in Theorem \ref{omega3} such that
%   \begin{equation}\label{epsilon1}
%     \epsilon<\min\{(\eta_0-\delta_0)/3,\nu\}
%   \end{equation}
% and fix it, then let
%   \[L:=e^{(\eta_0-\delta_0-\epsilon)}>1\]
% Pick $K$ big enough in Theorem \ref{omega3} to be such that
%   \[(3M)^{\frac{1}{K}}<L\]
% say, there exists $ \sigma>0$,
% \begin{equation}\label{K}
% (3M)^{\frac{1}{K}}\leq L-\sigma<L
% \end{equation}
Then for the left hand side of \eqref{last}, there are three cases:
\begin{enumerate}
  \item both $|-n-y_{1}|>\frac{n}{K}$ and $|n-y_{2}|>\frac{n}{K}$
  \item one of them is large, say $|-n-y_{1}|>\frac{n}{K}$ while $|n-y_{2}|\leq\frac{n}{K}$
  \item both small.
\end{enumerate}

For $(1)$,
\[
\frac{1}{2n}e^{(\eta_0-\delta_0+\gamma(E_j)-\epsilon)(2n+1)}\leq e^{2n(\gamma(E_j)+\epsilon)}
\]
Since by our choice \eqref{epsilon1},
 $\eta_0-\delta_0+\gamma(E_j)-\epsilon>\gamma(E_j)+\epsilon$, for $n$ large enough, we get a contradiction.

For $(2)$,
\[
  % \begin{aligned}
    \frac{1}{2n}e^{(\eta_0-\delta_0+\gamma(E_j)-\epsilon)(2n+1)}
    \leq e^{(\gamma(E_j)+\epsilon)(2n+1)}(M)^{\frac{n}{K}}
    % \frac{1}{2Cn}e^{(\eta_0-\delta_0-\epsilon)(2n+1)}&\leq e^{\epsilon(2n+1)} L^n\\
    % &\leq e^{\epsilon(2n+1)} e^{(\eta_0-\delta_0-\epsilon)n}\\
    % \frac{1}{2Cn}e^{(\eta_0-\delta_0-\epsilon)(n+1)}
    % &\leq e^{2\epsilon(n+1)}
  % \end{aligned}
\]
is in contradiction with \eqref{epsilon1} and \eqref{K}

For $(3)$, with \eqref{epsilon1} and \eqref{K}
\[
% \begin{aligned}
  \frac{1}{2n}e^{(\eta_0-\delta_0+\gamma(E_j)-\epsilon)(2n+1)}\leq M^{\frac{2n}{K}}\leq (L-\sigma)^{2n}\leq(e^{(\eta_0-\delta_0+\gamma(E_j)-\epsilon)}-\sigma)^{2n},
% \end{aligned}
\]
also a contradiction.

Thus our assumption that $2n+1$ is not  $(\gamma(\tilde{E})-8\epsilon_0,n,\tilde{E},\tilde{\omega})$-regular is false. Theorem \ref{thm2} follows.
\end{proof}

Note that we have established the following more precise version of
Theorem \ref{thm2}

\begin{thm}\label{thm22}
  There exists $ \Omega_0$ with $\mathbb{P}(\Omega_0)=1$, such that
  for every $ \tilde{\omega}\in\Omega_0$, for any g.e. $\tilde{E}$ of
  $H_{\tilde{\omega}}$, and $\epsilon>0,$ there exists $
  N=N(\tilde{E},\tilde{\omega}, \epsilon)$, such that for every $ n>N$, $2n,~2n+1$ are $(\gamma(E)-\epsilon,n,\tilde{E},\tilde{\omega})$-regular.
\end{thm}

It is a standard patching argument (e.g. proof of Theorem 3 in
\cite{j}) that this implies $|\Psi_E(n)|\leq C_{E,\epsilon}e^{-(\gamma(E)
  -\epsilon)n}$ for any $\epsilon>0.$ Combined with Theorem \ref{CS1},
this immediately implies that we have Lyapunov behavior at every
generalized eigenvalue.

\begin{thm}\label{CS2}
  For a.e. $\omega$  for all generalized eigenvalues $E$, we have
  \begin{equation}\label{ld}
    % \begin{aligned}
   \lim_{n\to\infty} \frac{\log\Vert T_{[0,n],E,\omega}\Vert}{n+1}=\gamma(E)
  \end{equation}

\end{thm}

\section{Uniform and Quantitative Craig-Simon}\label{uniformcs}
Craig-Simon theorem \ref{CS1}  implies that for a.e. $\omega$ and
every $E\in\sigma$ there exists $N(\omega,E)$ such that for $n>N,$
$\Vert T_{[0,n],E,\omega}\Vert\leq e^{(n+1)(\gamma(E)+\epsilon)}.$ For the proof
of dynamical localization one however needs a statement of this type
with $N$ uniform in $E.$ Such a statement is the goal of this
section. We will show that it holds for any ergodic dynamical system
satisfying the uniform LDT (Large Deviation Type) condition: Lemma \ref{ldt lemma}. Thus this
result has more general nature than the rest of the paper and may be
of independent interest. In particular, it is applicable to
quasiperiodic dynamics with Diophantine frequencies and analytic
sampling functions. We note that uniform LDT condition can also be
replaced by a combination of a pointwise LDT condition and continuity
of the Lyapunov exponent.

We have:
\begin{thm}\label{QCS}
  Let the ergodic family $H_\omega$ satisfy Lemma \ref{ldt lemma}. Fix
  $\epsilon_0>0$. For a.e. $\omega$ (we denote this set as
  $\Omega_2=\Omega_2(\epsilon_0)$), there exists $N_2(\omega)$, such
  that for any $n>N_2(\omega)$, $E\in \sigma $,
  \[
    |P_{[0,n],E,\omega}|\leq e^{(\gamma(E)+\epsilon_0)(n+1)}
  \]
\end{thm}
An immediate corollary is
\begin{cor}
Let $H_\omega,\epsilon_0$ be as above. Then there exists $\Omega_2$ with $\mathbb{P}(\Omega_2)=1$, such that for $\omega\in\Omega_2$, there exists $N_2(\omega)$ such that
  \[
  \max\left\{|P_{[0,n],E,\omega}|,|P_{[-n,0],E,\omega}|,|P_{[n+1,2n+1],E,\omega}|,|P_{[2n+1,3n+1],E,\omega}|\right\}\leq e^{(\gamma(E)+3\epsilon_0)(n+1)}.
  \]
\end{cor}
 Thus we can replace Corollary \ref{omega2} with this uniform version.
% \begin{lemma}
%   If $Q(x)$ is a polynomial of degree $n$, and $x_1,\cdots,x_n$ are $n$ uniformly distributed points in $[x_1,x_n]$. If $Q(x_i)\leq a$ for any $i=1,\cdots,n$, then $Q(x)\leq an^c$ for some $c>0$ and any $x\in[x_1,x_n]$.
% \end{lemma}
\begin{proof}
  We start with the following
  %
  % We begin with an elementary Lemma:
  \begin{lemma}\label{elementary}
   Let $Q(x)$ be a polynomial of degree $n-1$. Let    $x_i=\cos{\frac{2\pi(i+\theta)}{n}}$, $0<\theta<1/2$, $i=1,2,\cdots,n$. If $Q(x_i)\leq a^n$, for all $i$, then $Q(x)\leq Cna^n$, for all $x\in[-1,1]$, where $C=C(\theta)$ is a constant.
  \end{lemma}
  \begin{proof}
    By Lagrange interpolation, we have
    \begin{equation}
      Q(x)=\sum_{i=1}^{n}Q(x_i)\prod_{j\neq i}\frac{x-x_j}{x_i-x_j}
    \end{equation}
    Note that
    \[
    \sum_{j\neq i}\ln|x_i-x_j|=\sum_{j\neq i}\left\{\ln\left|\sin\frac{\pi(i+j+2\theta)}{n}\right|+\ln\left|\sin\frac{\pi(i-j)}{n}\right|+\ln 2\right\}=: A+B+(n-1)\ln 2.
    \]
    We will use the following lemma without giving a proof.
    \begin{lemma}[Lemma 9.6 in \cite{avila2009ten}]\label{9.6}
      Let $p$ and $q$ be relatively prime. Let $1\leq k_0\leq q$ be such that
      \[
      \vert \sin 2\pi(x+k_0p/(2q))\vert  =\min
      _{1\leq k\leq q}\vert \sin 2\pi(x+kp/(2q))\vert.
      \]
      Then
      \begin{equation}
        \ln q+\ln (2/\pi)<\sum_{\substack{k=1\\ k\neq k_0}}^{q}\ln\vert\sin 2\pi(x+kp/(2q))\vert+(q-1)\ln 2\leq \ln q.
      \end{equation}
    \end{lemma}

 For $B$, we take $p=1,~q=n,~x=-i/(2n),~k=j$. Then $k_0=i$, and we get
    \[
      B\geq \ln n+\ln (2/\pi)-(n-1)\ln 2.
    \]
For $A$, we estimate by Lemma \ref{9.6} with $p=1,~q=n,~x=(i+2\theta)/2n,~k=j$. If $k_0=j_0$ is the minimum term of $\ln|\sin\frac{\pi(i+j+2\theta)}{n}|$, then
    \[
      A\geq \ln n+\ln (2/\pi) -(n-1)\ln 2-\ln\left|\sin\frac{\pi(2i+2\theta)}{n}\right|+\ln\left|\sin\frac{\pi(i+j_0+2\theta)}{n}\right|
    \]
    For $0<\theta<1/4$, we have
    \[
    \frac{|\sin\frac{\pi(2i+2\theta)}{n}|}{|\sin\frac{\pi(i+j_0+2\theta)}{n}|}
    = \frac{|\sin\frac{\pi(2i+2\theta)}{n}|}{|\sin\frac{\pi\cdot 2\theta}{n}|}\leq \frac{1}{|\sin\frac{\pi\cdot 2\theta}{n}|}= O(n)
    \]
   Thus
    \[
      \sum_{j\neq i}\ln|x_i-x_j|\geq -(n-1)\ln 2+\ln n+C
    \]
    Writing $x=\cos\frac{2\pi a}{n}$, by Lemma \ref{9.6}, we get
    \[
      \sum_{j\neq i}\ln |x-x_j|\leq -(n-1)\ln 2+ 2\ln n+ C
    \]
    Thus
    \[
    \prod_{j\neq i}\frac{x-x_j}{x_i-x_j}\leq Cn
    \]
    and we have
    \[
    Q(x)\leq Cna^n
    \]
  \end{proof}
Now we can finish the proof of Theorem \ref{QCS}.

We know that $\sigma$ is compact, so contained in some bounded closed
interval. Assume we are dealing with $[a,a+A]$. Unifrom LDT implies
that $\gamma$ is a continuous function of $E$ \cite{DK}. Since $\gamma(E)$ is uniformly continuous,  for any $\epsilon_0$, there exists $\delta_0$ such that
\begin{equation}\label{holder}
|\gamma(E_x)-\gamma(E_y)|\leq \epsilon_0,\quad if~ |E_x-E_y|\leq \delta_0.
\end{equation}

Divide the interval $[a,a+A]$ into length-$\delta_0$ sub-intervals. There are $K=[A/\delta_0]+1$ of them (the last one may be shorter). Denote them as $I_k$, for $k=1,\cdots, K$. For $I_k=[E_{k,n},E_{k+1,n}]$, let $E_{k1,n},\cdots,E_{kn,n}$ be distributed as in Lemma 5.3. Namely, set $ E_{ki,n}=E_{k,n}+(x_i+1)\delta_0/2$,  where $x_i$ are as in Lemma 5.3, $0<\theta<1/2$. Note that for any $E_x$, $E_y\in [E_{k1,n},E_{kn,n}]$, $|\gamma(E_x)-\gamma(E_y)|\leq \epsilon_0$.
Since by the uniform-LDT condition
\[
\mathbb{P}\left(\left\{\omega:  \exists i=1,\cdots,n,~s.t.~|P_{[0,n],E_{ki,n},\omega}|\geq e^{(\gamma(E_{ki,n})+\epsilon_0)(n+1)} \right\}\right)\leq ne^{-\eta_0(n+1)},
\]
by Borel-Cantelli, for a.e. $\omega$, (we denote this set as $\Omega(k)$), there exists $N(k, \omega)$, such that for all $n>N(k,\omega)$,
\[
|P_{[0,n],E_{ki,n},\omega}|\leq e^{(\gamma(E_{ki,n})+\epsilon_0)(n+1)},\quad \forall i=1,\cdots,n.
\]
If we denote $\gamma_{k,n}=\inf_{E\in [E_{k1,n},E_{kn,n}]}{\gamma(E)}$, then by \eqref{holder}
\[
|P_{[0,n],E_{ki,n},\omega}|\leq e^{(\gamma(E_{ki,n})+\epsilon_0)(n+1)}\leq e^{(\gamma_{k,n}+2\epsilon_0)(n+1)}, \quad \forall i=1,\cdots,n.
\]

Let $M$ be big enough such that, for any $n>M$, $n^c\leq e^{\epsilon_0(n+1)}$. Thus by Lemma \ref{elementary}, applied to $Q(x)=P(E_{k,n}+\frac{(x+1)\delta_0}{2})$, for $E\in[E_{k,n},E_{k+1,n}]$, $n>\max\{N(k,\omega),M\}$,
\[
|P_{[0,n],E,\omega}|\leq n^ce^{(\gamma_{k,n}+2\epsilon_0)(n+1)}\leq n^ce^{(\gamma(E)+2\epsilon_0)(n+1)}\leq e^{(\gamma(E)+3\epsilon_0)(n+1)}
\]
Let $\Omega_2=\bigcap\limits_k \Omega(k)$,  $\tilde{N}(\omega)=\max_k\{N(k,\omega),M\}$. Then for any $n>\tilde{N}(\omega)$,
\[
|P_{[0,n],E,\omega}|\leq e^{(\gamma(E)+3\epsilon_0)(n+1)},\quad \forall E\in [a,a+A]
\]
\end{proof}
This allows us to also obtain a quantitative version of Theorem \ref{QCS}. Assume the $N_2(\omega)$ in Theorem \ref{QCS} is chosen to be the smallest satisfying the condition. Let $l\in\mathbb{Z}$, $N_2(l,\omega)=N_2(T^l\omega)$. Let $\bar{\Omega}_2=\bigcap_{l\in\mathbb{Z}}T^l\Omega_2$.
\begin{lemma}\label{5.5}
  For a.e. $\omega$ (we denote this set as $\tilde{\Omega}_2$), there exists $L_2=L_2(\omega)$, such that for all $\vert l\vert>L_2$, $N_2(l,\omega)\leq \ln^2\vert l\vert$. In particular, if $n>\ln^2 \vert l\vert$, then
  \[
  \vert P_{[l,l+n],E,\omega}\vert\leq e^{(\gamma(E)+\epsilon_0)(n+1)},~for~all~E\in\sigma
  \]
\end{lemma}
\begin{proof}
 Let $\omega\in\bar{\Omega}_2$, $l\in\mathbb{Z}$, $k\in\mathbb{N}$. By Theorem \ref{QCS}, $\bar{\Omega}$ has full measure. We have
 \begin{eqnarray*}
  \mathbb{P}\{\omega:N_2(l,\omega)\geq k\} & \leq &\sum_{n=k}^{\infty}\mathbb{P}\{\omega:N_2(l,\omega)=n\}\leq \sum_{n=k}^{\infty}\mathbb{P}(B^+_{[l,l+n-1],E})\\
  & \leq &\sum_{n=k}^{\infty}Ce^{-(\gamma(E)+\epsilon_0)n}\leq Ce^{-(\gamma(E)+\epsilon_0)k}
\end{eqnarray*}

  Thus \[\mathbb{P}\{\omega:N_2(l,\omega)\geq \ln^2\vert l\vert\}\leq Ce^{-(\gamma(E)+\epsilon_0)(\ln^2\vert l\vert)}\]
  By Borel-Cantelli lemma, we get the result and the corresponding $\tilde{\Omega}_2$.
\end{proof}

\section{Dynamical Localization}\label{dyn}
%
% \begin{remark}\label{N1}
%   Note that we can actually shift the operator and use center point $l$ instead of $0$. Then we will get $\Omega_1(l)$ instead of $\Omega_1$, $N_1(l,\omega)$ instead of $N_1(\omega)$. And if we pick $N_1(l,\omega)$ in the theorem as the smallest interger satisfying the conclusion, we can estimate when will $N_1(l,\omega)\leq \ln^2 |l|$, which is very useful in the proof for dynamical localization in section 6.
% \end{remark}
% \begin{remark}\label{N2}
%   Similar as remark \ref{N1} and \ref{N3}, we can get $\Omega_2(l)$, $N_2(l,\omega)$ instead. Note $M$ is independent of $l$, and we can then estimate in the same way that, for $a.e.\omega$, (We denote this set as $\Omega_{N_2}$), there exists $L_2=L_2(\omega)$, such that for any $|l|>L_2$, $N_2(l,\omega)\leq \ln^2 |l|$
% \end{remark}
%   \begin{remark}\label{N3}
%    Similar to remark \ref{N1}, we can get $\Omega_3(l)$, $N_3(l,\omega)$ instead, and get the result that for $a.e.\omega$, (We denote this set as $\Omega_{N_3}$,)  there exists $L_3(\omega)$, such that for any $|l|>L_3$, $N_3(l,\omega)\leq \ln^2 |l|$ in section 6, too.
%   \end{remark}

Now we have established the spectral localization for 1-d Anderson model. With some more effort, we can get the dynamical localization.
We say that $H_\omega$ exhibits dynamical localization if for $a.e.~\omega$, for any $\epsilon>0$, there exists $\alpha=\alpha(\omega)>0$, $C=C(\epsilon,\omega)$, such that for all $x,y\in\mathbb{Z}$:
  \[
    \sup_t |\langle\delta_x,e^{-itH_{\omega}}\delta_y\rangle|\leq C_\epsilon e^{\epsilon|y|}e^{-\alpha|x-y|}
  \]
According to \cite{del1996operators}, we only need to prove that for a.e. $\omega$, $H_\omega$ has SULE (Semi-Uniformly Localized Eigenfunction). We say $H$ has SULE if $H$ has a complete set $\{\varphi_E\}$ of orthonormal eigenfunctions, such that there is $\alpha>0$, and for each $\epsilon>0$, a $C_\epsilon$ such that for any eigenvalue $E$, there exists $l=l_E\in\mathbb{Z}$, such that
  \[
    |\varphi_E(x)|\leq C_\epsilon e^{\epsilon|l_E|}e^{-\alpha|x-l_E|},\quad x\in\mathbb{Z}
  \]
  In fact, we will prove that $\vert \varphi_E(x)\vert\leq C_\epsilon e^{C\ln^2(1+\vert l_E\vert)}e^{-\alpha\vert x-l_E\vert}$,  see
   \eqref{*1}, \eqref{*2}.
In order to do this, we need to modify Lemma \ref{omega1}, Lemma \ref{omega3} using the same method as in Lemma \ref{5.5}.
Assume the $N_i(\omega),~i=1,3$ in Lemmas \ref{omega1}, \ref{omega3} are chosen to be the smallest parameters satisfying the condition. Let $l\in\mathbb{Z}$, $N_i(l,\omega)=N_i(T^l\omega)$. Let $\bar{\Omega}_i=\bigcap_{l\in\mathbb{Z}}T^l\Omega_i$, $i=1,3$.
\begin{lemma}\label{lem}
  For a.e. $\omega$ (we denote this set as $\tilde{\Omega}_{1,3}$),  there are $L_1(\omega),L_3(\omega)$ such that for any $\vert l\vert>\max\{L_1,L_3\}$, \[\max\{N_1(l,\omega),N_3(l,\omega)\}\leq\ln^2\vert l\vert
\]
\end{lemma}
\begin{proof}
  Let $\omega\in\bar{\Omega}_1$, $l\in\mathbb{Z}$,
   $k\in \mathbb{N}$, then by \eqref{A5}
  \[
  \mathbb{P}\{\omega:N_1(l,\omega)>k\}\leq \sum_{n= k}^\infty\mathbb{P}(\Omega_{\delta,n,\pm})\leq\sum_{n=k}^\infty 2m(\sigma)e^{-\delta_0(2n+1)}\leq Ce^{-\delta_0(2k+1)}
  \]
Thus
\[
\mathbb{P}\{\omega:N_1(l,\omega)>\ln^2 \vert l\vert\}\leq Ce^{-\delta_0(2\ln^2\vert l\vert+2)}
\]
By Borel-Cantelli lemma, we can get the result. The same argument works for $N_3$.
\end{proof}
Then we rebuild the criteria for regularity around a singular point $l$.
\begin{lemma}\label{last cor}
   For a.e. $\omega$ (we denote this set as $\tilde{\Omega}$), for any $l$, there exists $N(l,\omega)$,
   such that for any $n>N(l,\omega)$ and for all $E\in \sigma$
   either $l$ or $l+2n+1$, and either $l$ or $l-2n-1$ are $(\gamma(E)-8\epsilon_0,n,E,\omega)$-regular.
\end{lemma}
\begin{proof}
  In section 4, we proved that either $0$ or $2n+1$ is $(\gamma(E)-8\epsilon_0,n,E,\omega)$-singular for all  $n>N(\omega)$, with $N(\omega)=\max\{N_1(\omega),N_2(\omega),N_3(\omega)\}$. Here we set $N(l,\omega)=\max\{N(T^l\omega),N(T^{-l}\omega)\}$, and modify $\tilde{\Omega}$ accordingly.
\end{proof}

Now, take $\tilde{\Omega}=\tilde{\Omega}_2\cup\tilde{\Omega}_{1,3}$ and fix $\omega\in\tilde{\Omega}$. We omit $\omega$ from notations from now on.

By Lemma \ref{lem} and Lemma \ref{5.5}, there exist $L_1$, $L_2$, $L_3$ such that for all $|l|>\max\{L_1,L_2,L_3\}$,
  \[
    N_i(l)\leq \ln^2 |l|,\quad \forall i=1,2,3
  \]
  for all $E\in\sigma$.

Let $l_E$ be a position of the maximum point of $\varphi_E$.
Take $L_4$ with $\ln^2L_4\geq[\frac{\ln 2}{\gamma(E)-8\epsilon_0}]+1$. For any $n\geq \ln^2L_4$, and any e.v. $E$, $l_E$ is naturally $(\mu-8\epsilon_0,n,E)$-singular by \eqref{possion}.

Let $L=\max\{L_1, L_2, L_3, L_4\}$, $N(l):=\max\{N_1(l),N_2(l),N_3(l),\frac{\ln 2}{\gamma(E)-8\epsilon_0}\}$. Then for any $|l|>L$,
\begin{equation}\label{Nl}
  N(l)\leq \ln^2 |l|
\end{equation}

If $|l_E|>L$, then for any $n\geq N(l_E)$, $l_E$ is $(\gamma(E)-8\epsilon_0,n,E)$-singular, so $x=l_E\pm(2n+1)$ is $(\gamma(E)-8\epsilon_0,n,E)$-regular. By \eqref{possion}, for any $|x-l_E|\geq N(l_E)$
  \begin{equation}\label{n>}
    |\varphi_E(x)|\leq 2e^{-(\gamma(E)-8\epsilon_0)|x-l_E|}
  \end{equation}
Since $\varphi_E$ is normalized, in fact for all $x$,
\[
  |\varphi_E(x)|\leq 2e^{(\gamma(E)-8\epsilon_0)N(l_E)}e^{-(\gamma(E)-8\epsilon_0)|x-l_E|}
\]
By \eqref{Nl}, for any $\epsilon$,
\begin{equation}\label{*1}
  |\varphi_E(x)|\leq 2e^{(\gamma(E)-8\epsilon_0)\ln^2 (1+|l_E|)} e^{-(\gamma(E)-8\epsilon_0)|x-l_E|}
\end{equation}

If $|l_E|\leq L$, for any $\epsilon$, for $n\geq N(l_E)$, we use the same argument as \eqref{n>} and get
\begin{equation}\label{*3}
  |\varphi_E(x)|\leq 2e^{-(\gamma(E) -8\epsilon_0)|x-l_E|}\leq 2e^{\epsilon \ln^2 (1+|l_E|)}e^{-(\gamma(E)-8\epsilon_0)|x-l_E|}
\end{equation}
 While for $n\leq N_{l_E}$, set $M_{2\epsilon}=\min_{k\in[-L,L],~|x-k|<N(k)}
 \{e^{\epsilon \ln^2(1+|k|)} e^{-(\gamma(E)-8\epsilon_0)|x-k|}\}$ and  $C_{2\epsilon}=M_{2\epsilon}^{-1}$. Then for all  $|x-l_E|< N(l_E)$,
\begin{equation}\label{*2}
  |\varphi_E(x)|\leq 1\leq C_{2\epsilon} e^{\epsilon \ln^2 (1+|l_E|)} e^{-(\gamma(E)-8\epsilon_0)|x-l_E|}
\end{equation}

Thus for $C_\epsilon=\max\{2,C_{2\epsilon}\}$, \eqref{*1} \eqref{*3}
and \eqref{*2} provide SULE. \qed

 % If $l$ is $(\mu-8\epsilon_0,n,E,\omega)$-singular, then $x$ is $(\mu-8\epsilon_0,n,E,\omega)$-regular.
 %
 % Let $\varphi_E$ be the normalized eigenfunction of $E$. For any $|x-l|=2n+1$, $n>N_{l,E,\omega}$, $ |\varphi_E(x)|\leq e^{-C|x-l|}$, for some $C=>0$.

 % \begin{remark}\label{N3}
 %  Similar to remark \ref{N1}, we can get $\Omega_3(l)$, $N_3(l,\omega)$ instead, and get that for $a.e.\omega$, (We denote this set as $\Omega_{N_3}$,)  there exists $L_3(\omega)$, such that for any $|l|>L_3$, $N_3(l,\omega)\leq \ln^2 |l|$.
 % \end{remark}

 % \begin{remark}\label{N1}
 %   Note that we can actually shift the operator and use center point $l$ instead of $0$. Then we will get $\Omega_1(l)$ instead of $\Omega_1$, $N_1(l,\omega)$ instead of $N_1(\omega)$. And if we pick $N_1(l,\omega)$ in the theorem as the smallest interger satisfying the conclusion, we can estimate when will $N_1(l,\omega)\leq \ln^2 |l|$, which is very useful in the proof for dynamical localization in section 6.
 %   In fact, $\mathbb{P}\{\omega: N_1(l,\omega)>\ln^2 |l|\}\leq C' e^{-\delta_0(2|\ln^2|l|+1)}$, By Borel-Cantelli, for $a.e.\omega$,(We denote this set as $\Omega_{N_1}$,) there exists $L_1(\omega)$, such that for any $|l|>L_1(\omega)$, $N_1(l,\omega)\leq \ln^2|l|$.
 % \end{remark}

\section*{Acknowledgments}
This research was partially supported by the NSF DMS-1401204. X. Z. is grateful to Wencai Liu for inspiring thoughts and comments for Sec. 5. We also thank Barry Simon for his encouragement.

% \cite{simon1982schrodinger}
% \cite{von1989new}
% \cite{jitomirskaya1999metal}
% \cite{cycon2009schrodinger}
% \cite{bucaj2017localization}
% \cite{craig1983subharmonicity}
% \cite{tsay1999some}
% \cite{avbila2009ten}
% \cite{duarte2016lyapunov}
%\cite{goldstein2001holder}
%\cite{bucaj2017localization}

\bibliographystyle{unsrt}
\bibliography{mybib1}

\end{document}